\documentclass[11pt,a4paper]{article}
\usepackage{graphicx,color}
\usepackage{amsmath,amssymb,dsfont,enumerate,epsfig,multirow,longtable,url,bm}
\usepackage{epsfig,epstopdf,hhline,a4wide}
\usepackage{cases}
\usepackage{algorithm}
\usepackage{algorithmic}

\usepackage[numbers]{natbib}
\usepackage{pgfplots}
\usetikzlibrary{arrows,positioning,decorations.pathreplacing,shapes}

\usepackage{thmtools}
\usepackage{thm-restate}

\newenvironment{proof}{\noindent\emph{Proof\ }}{\hspace*{\fill}$\Box$\medskip}

\newtheorem{theorem}{Theorem}
\newtheorem{definition}{Definition}
\newtheorem{lemma}{Lemma}

\newtheorem{proposition}{Proposition}

\usepackage{amsmath}
\usepackage{paralist}
\usepackage{framed}

\newcommand\restr[2]{{% we make the whole thing an ordinary symbol
  \left.\kern-\nulldelimiterspace % automatically resize the bar with \right
  #1 % the function
  \vphantom{\big|} % pretend it's a little taller at normal size
  \right|_{#2} % this is the delimiter
  }}

\newcommand{\pred}{\texttt{pred}}
\newcommand{\vect}[1]{\ensuremath{\bm{#1}}}

\usepackage[utf8]{inputenc} % Required for inputting international characters
\usepackage[T5]{fontenc} % Output font encoding for international characters

\usepackage{color, colortbl}
\usepackage{hyperref}
\hypersetup{colorlinks,
            linkcolor=blue,
            citecolor=blue,
            urlcolor=magenta,
            linktocpage,
            plainpages=false}
\usepackage[capitalise,noabbrev]{cleveref}

\title{Online Primal-Dual Algorithms with Predictions for Packing Problems}

\author{
Nguyễn Kim Thắng \\
University Paris-Saclay, IBISC, France
\and Christoph D\"urr \\
Sorbonne University, CNRS, LIP6, France
}

\begin{document}

\maketitle

\begin{abstract}
The domain of online algorithms with predictions has been extensively studied
for different applications such as scheduling, caching (paging), clustering, ski rental, etc.
Recently, Bamas et al., aiming for an unified method, have provided a primal-dual framework 
for linear covering problems. They extended the online primal-dual method by incorporating 
predictions in order to achieve a performance beyond the worst-case case analysis.
In this paper, we consider this research line and  present a framework to design algorithms with predictions for non-linear 
packing problems. We illustrate the applicability of our framework in submodular maximization
and in particular ad-auction maximization in which the optimal bound is given and supporting experiments are provided.
\end{abstract}

%\thispagestyle{empty}

%\newpage

%\setcounter{page}{1}

%!TEX root = main.tex

\section{Introduction}

%% short online algorithm

Online computation \cite{BorodinEl-Yaniv05:Online-computation} is a
well-established field in theoretical computer science. In online
computation, inputs are released one-by-one in form of a request sequence.
Requests arrive over time and each request must be answered by the algorithm
with an irrevocable action, without any information about future requests to
come. These actions define a solution to the problem at hand, and the
algorithm wants to optimize its objective value. Usually this value is
normalized by the optimum, which could have been computed if all the requests
were known in beforehand.  The resulting \emph{competitive ratio} measures
the price of not knowing the future requests.

%% machine learning as a way to go beyond the worst case
Several models have been studied, which relax some of the assumptions in the standard online computation model, see  \cite{Roughgarden19:Beyond-worst-case,Roughgarden20:Beyond-the-Worst-Case} for a recent survey.

%% recent trends
One of the models is motivated by spectacular advances in machine learning (ML). 
In particular, the capability of ML methods in predicting patterns of future requests would provide useful information to be exploited by online 
algorithms. A general framework for incorporating ML predictions into algorithm design in order to achieve a better performance than the worst-case one is formally introduced by \cite{LykourisVassilvtiskii18:Competitive-caching}. 
This is rapidly followed by work studying online algorithms with predictions \cite{MitzenmacherVassilvitskii20:Beyond-the-Worst-Case} 
in a large spectrum of problems 
such as scheduling \cite{LattanziLavastida20:Online-scheduling,Mitzenmacher20:Scheduling-with}, 
caching (paging) \cite{LykourisVassilvtiskii18:Competitive-caching,Rohatgi20:Near-optimal-bounds,AntoniadisCoester20:Online-metric}, 
ski rental \cite{GollapudiPanigrahi19:Online-algorithms,KumarPurohit18:Improving-online,AngelopoulosDurr20:Online-Computation}, etc.  In contrast to online computation with advice (see \cite{Komm16:Introduction-to-Online} for a survey), the aforementioned models handle predictions in a careful manner, guaranteeing good performance even if the predictions is completely wrong.
More precisely, on the one hand, if the predictions are accurate, one expects to achieve a good performance 
as in the offline setting where the input are completely given in advance. On the other hand, if the predictions are misleading, one still has to maintain a competitive solution as in the online setting where no predictive information is available.

%% unified methods/primal-dual 
This issue has been addressed (for example the aforementioned work) 
by subtly incorporating predictions and exploiting specific problems' structures.
The primal-dual method is an elegant and powerful technique for the design of algorithms \cite{WilliamsonShmoys11:The-design-of-approximation}, 
especially for online algorithms \cite{BuchbinderNaor09:Online-primal-dual}.
In the objective of unifying previous ad-hoc approaches,  
\cite{BamasMaggiori20:The-Primal-Dual-method} presented a primal-dual 
approach to design 
online algorithms with predictions for covering problems with a linear objective function. In this paper, we build on  previous work and  
provide a primal-dual method with predictions for packing problems with non-linear objectives;
following a research direction and answering some open questions mentioned  in \cite{BamasMaggiori20:The-Primal-Dual-method}.

%The first explicite primal-dual approach in online algorithms
%is given by Buchbinder and Naor [34] who presented general algorithms based on
%the multiplicative weights update method for covering/packing problems. 
%Several uses of the primal-dual method and competitive algorithms for online algorithms then followed; for example the generalized caching problem [17, 16], the ad-auction maximization [35] and others.  

%We measure the performance of an algorithm by the \emph{competitive ratio}. Specifically, 
%an algorithm is $r$-competitive if for any instance, the ratio between the cost of the algorithm and that of 
%an optimal solution is at most $r$.  

\subsection{Problem and Model}  	\label{sec:model}
 
\paragraph{Packing problem} In the packing problem studied in this paper we are given of a sequence of $n$ elements ${\cal E}$, and $m$ unit capacity resources, numbered from $1$ to $m$.  Each element $e\in\cal E$ has a profit $w_e$ and a resource requirement $b_{i,e}\geq 0$ for each resource $i$. The goal is to select a set of elements of maximum total profit, such that no resource capacity is exceeded.  Formally the goal is to chose a vector $\vect{x} \in\{0,1\}^{|\cal E|}$ maximizing $\sum_{e\in\cal E} w_e x_e$ such that for every resource $i\in\{1,\ldots,m\}$ we have $\sum_{e\in\cal E} b_{i,e} x_e \leq 1$.  

\paragraph{Fractional variant}
In the fractional variant of this problem, the integrality constraint on $x$ is removed, and the goal is to optimize over vectors $x\in[0,1]^{|\cal E|}$.
By $I$ we denote the instance of a packing problem, which consists of $\langle {\cal E}, (w_e)_{e\in \cal E}, m, (b_{i,e})_{1\leq i\leq m, e\in\cal E}\rangle$.

\paragraph{Non-linear variant} We consider a generalization of this problem, where the objective is not to maximize the sum $\sum_{e\in\cal E} w_e x_e$, but $f(\vect{x})$ for a given function $f:\{0,1\}^{|\cal E|} \rightarrow {\mathbb R}^+$.  It is assumed that $f$ is monotone, in the sense that $f(\textbf{x})\leq f(\textbf{y})$, whenever $\textbf{y}$ equals $\textbf{x}$ except for a single element $e$ with $x_e=0$ and $y_e=1$. Describing $f$ explicitly needs space $O(2^{|{\cal E}|})$, hence it is assumed that $f$ is given as a black-box oracle.  For the fractional variant, the multilinear extension $F:[0,1]^{|\cal E|} \rightarrow {\mathbb R}^+$ is used instead of $f$, which is defined as
\[
    F(\vect{x}) = \sum_{S\subseteq \cal E} \prod_{e\in S} x_e \prod_{e\not\in S}(1-x_e) f(\vect{1}_S),
\]  
where $\vect{1}_S$ is the characteristic vector of the set $S$. 
Alternatively, $F(\vect{x}) = \mathbb{E} \bigl[ f(\vect{1}_{T})\bigr]$ where $T$ is a random set 
such that every element $e$ appears independently in $T$ with probability $x_{e}$. 
Note that $F(\vect{1}_{S}) = f(\vect{1}_{S})$ for every $S \subseteq \mathcal{E}$ --- $F$ is truly an extension of $f$ --- and if $f$ is monotone then so is $F$, in the sense that $F(\vect x)\leq F(\vect y)$ for all vectors $\vect x,\vect y\in[0,1]^{|\cal E|}$ such that $x_e \leq y_e$ for all elements $e$.

\paragraph{Online packing problem.} In the paper, we consider the model presented in \cite{BamasMaggiori20:The-Primal-Dual-method} 
(with several definitions rooted in \cite{LykourisVassilvtiskii18:Competitive-caching,KumarPurohit18:Improving-online}). In the online packing problem, elements $e\in\cal E$ arrive online. When element $e\in\cal E$ arrives, the algorithm learns its resource requirements $b_{i,e}$. At the arrival of $e$ only $x_e$ can be set by the algorithm, and cannot be changed later on.  Only $m$, the number of resources, are initially known to the algorithm.  The algorithm does not know $\cal E$ initially, not even its cardinality.  Formally an online algorithm $\cal A$ receives the elements of $\cal E$ in arbitrary order, and upon reception of element $e\in\cal E$ needs to choose a value for $x_e$.  At any moment, the vector $\vect x$ must satisfy all $m$ resource capacity constraints of the packing problem.  

We denote by ${\cal A}(I)$ the objective value of the solution produced by $A$ on an instance $I$.  A prediction is an integral value $x^{\pred}_{e}\in\{0,1\}$, which is given to the algorithm with every element $e$.  Confidence in the prediction is described by a parameter $\eta \in (0,1]$. For  convenience we choose larger $\eta$ value to represent smaller confidence.  Similarly to ${\cal A}(I)$, we denote by ${\cal P}(I)$ the objective value of the vector $\vect{x}^{\pred}$ produced by the prediction, and set ${\cal P}(I)=0$ whenever $\vect{x}^{\pred}$ does not satisfy all constraints. Finally we denote by ${\cal O}(I)$ the optimal fractional solution to the packing problem, which could have been computed if the whole instance $I$ were available from the beginning.

Algorithm $\mathcal{A}$ is 
 $c(\eta)$-\emph{consistent} and  $r(\eta)$-\emph{robust} if for every instance $I$, 
 \begin{align*}
 	\mathcal{A}(I) 	\geq 	\max\{c(\eta) \cdot \mathcal{P}(I), r(\eta) \cdot \mathcal{O}(I) \}.
 \end{align*}
Ideally we would like that $c(\eta)$ tends to $1$, when $\eta$ approaches $0$, meaning that with high confidence, the algorithm performs at least as good as the prediction. Additionally we would like that $r(\eta)$ tends to the same competitive ratio as in the standard online setting (without provided predictions), when $\eta$ approaches $1$.
  
\paragraph{Objective function}
At any moment in time, the algorithm maintains a vector $\vect{x}$ over the arrived elements. Conceptually this vector is padded with zeros to create a vector belonging to $[0,1]^{|\cal E|}$, which for  convenience we also denote by $\vect x$.  %We assume that the algorithm can compute the derivative $\nabla_{e} F(\vect{x})$ in constant time.
  
 %%%%%************************************
 %%%%%************************************
\subsection{Approach and Contribution}  
  
We use the primal-dual method in our approach.  
At a high level, it formulates a given problem as a (primal) linear program.  There is a corresponding dual linear program, which plays an important role in the method.  By weak duality, a solution to the dual bounds the optimal solution to the primal.  Every arrival of an element in the online problem, translates in the arrival of a variable in the primal linear program, and a corresponding constraint in the dual linear program. The algorithm updates the fractional solutions to both the primal and 
dual in order to maintain their feasibility. Then, the competitive ratio of the algorithm is established 
by showing that every time an update of primal and dual solutions is made, the increase of the primal 
can be bounded by that of the dual up to some desired factor.   
 
Note that, similar to \cite{BamasMaggiori20:The-Primal-Dual-method}, we focus on constructing fractional solution 
which is the main step in the primal-dual method. In order to derive algorithms for specific problems, online rounding schemes 
are needed and in many problems, such schemes already exist. We will provide references for the rounding schemes
tailor designed for different applications, but do not develop them in this paper.

We follow the primal-dual framework presented in \cite{Thang20:Online-Primal-Dual} to design competitive algorithms for fractional non-linear packing problems.
Given a function $f: \{0,1\}^{n} \rightarrow \mathbb{R}^{+}$, its \emph{multilinear extension} $F: [0,1]^{n} \rightarrow \mathbb{R}^{+}$ 
is defined as $F(\vect{x}) := \sum_{S \subseteq \mathcal{E}} \prod_{e \in S} x_{e} \prod_{e \notin S} (1 - x_{e}) \cdot f(\vect{1}_{S})$
where $\vect{1}_{S}$ is the characteristic vector of $S$ (i.e., the $e^{\textnormal{th}}$-component of 
$\vect{1}_{S}$ equals $1$ if $e \in S$ and equals 0 otherwise). 

The performance of an algorithm will depend on the following measures on the objective function.  

\begin{definition}[\cite{Thang20:Online-Primal-Dual}]	\label{def:max-local-smooth}
A differentiable function $F: [0,1]^{|\cal E|} \rightarrow \mathbb{R}^{+}$ is \emph{$(\lambda,\mu)$-locally-smooth} for some parameters $\lambda,\mu\geq 0$
if for any set $S \subseteq \mathcal{E}$, and for every vector $\vect{x} \in [0,1]^{|\cal E|}$, 
the following inequality holds:
$$
\langle \nabla F(\vect{x}), \vect{1}_{S} \rangle = 
\sum_{e \in S} \nabla_{e} F(\vect{x}) \geq \lambda F\bigl( \vect{1}_{S} \bigr) - \mu F\bigl( \vect{x} \bigr).
$$
where $\nabla_{e} F(\vect{x})$ denotes $\partial F(\vect{x})/\partial x_{e}$.
\end{definition}

Note that the $(\lambda,\mu)$-smoothness notion here differs from the usual notion of smoothness of functions in
convex optimization. The former, introduced in \cite{Thang20:Online-Primal-Dual}, is related to the definition of 
smooth games \cite{Roughgarden15:Intrinsic-Robustness}   
in the context of algorithmic game theory. The parameters $(\lambda,\mu)$ somehow describe how far 
the function is from being linear. For example a non-decreasing linear function is $(1,1)$-locally-smooth.  Another example is the multilinear extension of monotone submodular functions, which is $(1,2)$-locally-smooth \cite{Thang20:Online-Primal-Dual}.

The dual of a linear packing problem is a linear covering problem and vice-versa.  However this breaks when considering non-linear objective functions.  The usual trick is to linearize the model, by introducing auxiliary variables and work with what is called a configuration linear program. Now, how are we going to exploit the prediction? 
In our approach, we incorporate the predictive information in the primal-dual approach by modifying the coefficients of the constraints.  This in turn influences the increase of the primal variables and hence the solution produced by the algorithm.

This approach guarantees that:  
(i) if the confidence on the prediction is high ($\eta$ is close to 0) then value of variables corresponding to the elements selected by the prediction would get a large value; (ii) and inversely, if the confidence on the prediction is low ($\eta$ is close to 1) then
the variables would be constructed as in the standard online primal-dual method. 
The construction follows the multiplicative weights update based on the gradient of the multilinear extension \cite{Thang20:Online-Primal-Dual},
which generalizes the multiplicative update in \cite{BuchbinderNaor09:The-Design-of-Competitive,AzarBuchbinder16:Online-Algorithms}.  
Using the concept of local-smoothness, we show the feasibility of our primal/dual solutions (even when the prediction provides an infeasible one)  and also analyze the performance of the algorithm, as a function of locally-smoothness and confidence parameters.

\begin{restatable}{theorem}{Packing}
\label{thm:packing}
Let $F$ be the multilinear extension of the non-decreasing objective function $f$. Denote the \emph{row sparsity}
$d := \max_{i} |\{b_{ie}: b_{ie} > 0\}|$ and the \emph{divergence}
$\rho := \max_{i} \max_{e,e': b_{ie' > 0}} b_{ie} / b_{ie'}$.  
Assume that $F$ is $(\lambda, \mu)$-locally-smooth 
for some parameters $\lambda > 0$ and $\mu$. 
Then, for every $0 < \eta \leq 1$, 
there exists a $r(\eta)$-consistent and $r(\eta) \cdot \bigl( \frac{\lambda}{2\ln(1+ d\rho/\eta) + \mu} \bigr)$-robust algorithm
for the fractional packing problem where $r(\eta) = \min_{\vect{0} \leq \vect{u} \leq \vect{1}} F(\frac{\vect{u}}{1 + \eta} )/F(\vect{u})$.
\end{restatable}

In the following, we describe the applications of our framework to different classes of problems.
\begin{itemize}
\item \textbf{Linear objectives.} 
Numerous combinatorial problems can be formalized by linear programs with packing constraints
to which our framework applies.  
When  $f$ is  a monotone linear function, the smooth parameters of $F$ are
$\lambda = 1$ and $\mu = 1$. In this case, our algorithm guarantees the consistency $C(\eta) = \frac{1}{1 + \eta}$
and the robustness $R(\eta) =  \frac{\Omega(1)}{2(1 + \eta)\ln(1+ d\rho/\eta)}$. Note that in this case, the multilinear extension $F$ is rather easy to compute, as it satisfies $F(\vect x) = \sum_{e\in\cal E} x_e f(\vect 1_{\{e\}})$.
\item
\textbf{Submodular objectives.}
Submodular maximization constitutes a major research agenda and has been widely studied 
in optimization, machine learning \cite{Bachothers13:Learning-with,KrauseGolovin14:Submodular-Function,BianLevy17:Non-monotone-Continuous} 
and algorithm design \cite{FeldmanNaor11:A-unified-continuous,BuchbinderFeldman15:A-tight-linear}.
Using our framework, we derive an algorithm that yields a $(1 - \eta)$-consistent
and $\Omega \bigl( \frac{1}{2\ln(1+ d\rho/\eta)} \bigr)$-robust fractional solution for online submodular maximization 
with packing constraints. This performance is similar to that of linear functions up to constant factors. 
Note that using the online contention resolution rounding schemes \cite{FeldmanSvensson16:Online-contention},
generalizing the offline counterpart \cite{ChekuriVondrak14:Submodular-function}, 
one can obtain randomized algorithms for several specific constraint polytopes such as knapsack polytopes, 
matching polytopes and matroid polytopes. 
\end{itemize}

\paragraph{Ad-auctions problem.} 
An interesting particular packing problem is the ad-auctions revenue maximization.  
Informally, one needs to allocate ads in an online manner to advertisers with their budget constraints 
in order to maximize the total revenue from allocated ads. Any algorithm improving the ad allocation even by a small fraction 
would have a non-negligible impact. Building on the salient ideas of our framework, we give 
a $(1-\eta)$-consistent and $(1 - \eta e^{-\eta})$-robust algorithm for this problem. 
% This is the optimal trade-off between the consistency and robustness, because 
% any $(1 - \eta)$-consistent algorithm must have robustness at most $(1 - \eta e^{-\eta})$.  

\subsection{Related work}
The primal-dual method is a powerful tool for online computation.

A primal-dual framework for linear programs with packing/covering constraints was given in \cite{BuchbinderNaor09:The-Design-of-Competitive}.
Their method unifies several previous potential-function-based analyses and give a 
principled approach to design and analyze algorithms for problems with linear relaxations. 
A framework for covering/packing problems with 
convex/concave objectives whose gradients are monotone was provided by \cite{AzarBuchbinder16:Online-Algorithms}. Subsequently, 
\cite{Thang20:Online-Primal-Dual} presented algorithms dealing with non-convex functions 
and established the competitive ratio as a function of the smoothness
parameters of the objective function. The smoothness notion introduced in \cite{Thang20:Online-Primal-Dual} has rooted in smooth games 
defined by \cite{Roughgarden15:Intrinsic-Robustness} in the context of algorithmic game theory.

% Online algorithms with predictions
The domain of algorithms with predictions \cite{MitzenmacherVassilvitskii20:Beyond-the-Worst-Case}, 
or learning augmented algorithms, has recently emerged
and rapidly grown at the intersection of (discrete) algorithm design and machine learning (ML). The idea is to incorporate learning predictions,
together with ML techniques into algorithm design, in order to achieve performance guarantees beyond 
the worst-case analysis and provide specifically adapted solutions to different problems. 
Interesting results have been shown over a large spectrum of problems such as
 scheduling \cite{LattanziLavastida20:Online-scheduling,Mitzenmacher20:Scheduling-with}, 
 caching (paging) \cite{LykourisVassilvtiskii18:Competitive-caching,Rohatgi20:Near-optimal-bounds,AntoniadisCoester20:Online-metric}, 
 ski rental \cite{GollapudiPanigrahi19:Online-algorithms,KumarPurohit18:Improving-online,AngelopoulosDurr20:Online-Computation}, 
 counting sketches \cite{HsuIndyk19:Learning-Based-Frequency}, 
 bloom filters \cite{KraskaBeutel18:The-case-for-learned,Mitzenmacher18:A-model-for-learned}, etc. 
 Recently, \citet{BamasMaggiori20:The-Primal-Dual-method} have proposed a primal-dual approach to design online algorithms with predictions for linear problems with covering constraints. In this paper, by combining their ideas and the ideas from 
 \cite{BuchbinderNaor09:The-Design-of-Competitive,AzarBuchbinder16:Online-Algorithms,Thang20:Online-Primal-Dual},
 we present a primal-dual framework for more general problems with non-linear objectives and 
 packing/covering constraints (motivated by a question mentioned in \cite{BamasMaggiori20:The-Primal-Dual-method}). 
 
 % Matching
 Online matching and ad-auctions have been widely studied (see \cite{Mehta13:Online-Matching} and references therein). 
 Motivated by Internet advertising applications, 
 several works have considered the ad-auctions problem in various settings where forecast/prediction is available/learnable. 
 \citet{EsfandiariKorula18:Allocation-with} considered a model in which the input is stochastic and a forecast for the future items is given.       
The accuracy of the forecast is intuitively measured by the fraction of the value of an optimal solution that can be obtained from the stochastic input.
They provide algorithms with provable bounds in this setting. \citet{SchildVee19:Semi-Online-Bipartite} introduced 
the semi-online model in which unknown future has a predicted part and an 
adversarial part. They gave algorithms with competitive ratios depending on the fraction of the adversarial part (in the input). Closely related to our work is the model by \citet{MahdianNazerzadeh12:Online-Optimization}
in which given two algorithms, one needs to design a (new) algorithm which is robust to both given algorithms. They derived an algorithm for ad-auctions problems
that achieves a fraction of the maximum revenue of the given algorithms. The main difference to ours model is that their algorithm is \emph{not} robust in case a given algorithm provides infeasible solutions (which would happen for the prediction) whereas our algorithm is.

%TODO: add references to online matching with prediction

%\input{general}

%!TEX root = main.tex

\section{Primal-Dual Framework for Packing Problems}		\label{sec:packing}
%\paragraph{Packing Problem.}
%Let $\mathcal{E}$ be a set of $n$ resources 
%and let $f: \{0,1\}^{n} \rightarrow \mathbb{R}^{+}$ be an \emph{arbitrary} non-decreasing function.
%Let $x_{e} \in \{0,1\}$ be a variable indicating whether resource $e$ is selected. 
%The set of packing constraints $\sum_{e} b_{i,e} x_{e} \leq 1 ~\forall i$ (including $x_{e} \leq 1 ~\forall e$) are given in advance and 
%resources $e$ are revealed online one-by-one. At the arrival of resource $e$, one receives a prediction $x^{\pred}_{e} \in \{0,1\}$ 
%and needs to make a decision on $x_{e}$ while maintaining $\vect{x} = (x_{e})_{e \in \mathcal{E}}$ feasible to the set of constraints.
%The objective of the problem is to maximize $f(\vect{x})$. In this problem, we seek a fractional solution that 
%is consistent to the prediction and robust to the optimal offline solution. 
%
%
%
%
%\subsection{Algorithm for Fractional Packing}		\label{sec:packing-main}
%Recall that a differentiable function $F: [0,1]^{n} \rightarrow \mathbb{R}^{+}$ is $(\lambda,\mu)$-max-locally-smooth
%if for any set $S \subset \mathcal{E}$, and for every vector $\vect{x} \in [0,1]^{n}$, 
%the following inequality holds:
%$$
%\sum_{e \in S} \nabla_{e} F(\vect{x}) \geq \lambda F\bigl( \vect{1}_{S} \bigr) - \mu F\bigl( \vect{x} \bigr)
%$$
%%where $\vect{x} := \bigvee_{e \in S} \vect{x}^{e}$, meaning that $x_{e'}  = \max \{x^{e}_{e'}\}$ for any coordinate $e'$.
% 

\paragraph{Formulation.}
First, we model the packing problem as a configuration linear program. 
For the integral variant of the packing problem, the decision variable $x_{e}\in\{0,1\}$  indicates whether element $e$ is selected in the solution. A configuration is a set of elements $S \subseteq \mathcal{E}$ and could be feasible or not.  In addition to $\vect x$ the linear program contains variables $z_S\in\{0,1\}$ for every configuration $S$. 
The idea is that $z_S=1$ solely for the set $S$ containing all selected elements $e$, i.e.\ for which $x_e=1$.  In this case $S$ is feasible by the constraints imposed on $\vect x$.

The fractional variant of the packing problem is modeled with the same variables and constraints, but without the integrality constraints.  Now $x_e$ specifies the fraction with which $e$ is selected.  The $\vect z$ variables represent now a distribution on configurations $S$ and have to be consistent with $\vect{x}$ in the following sense. When $S$ is selected with probability $z_S$, then $e$ belongs to the selected set $S$ with probability $x_e$. Unlike for the integral linear program, $\vect z$ is not unique for given $\vect x$.  In our algorithm we chose a particular vector $\vect z$. Note that the support of $\vect z$ might contain non-feasible configurations.

We consider the following linear program and its  dual.  

\begin{minipage}[t]{0.45\textwidth}
\begin{align*}b
\max  \sum_{S} &f(\vect{1}_{S}) z_{S} \\
\sum_{e} b_{i,e}  \cdot x_{e}  &\leq 1 & &  \forall i & (\alpha_i)\\
\sum_{S: e \in S} z_{S}  &= x_{e} 	& & \forall e & (\beta_e)\\
\sum_{S} z_{S} &= 1 & & & (\gamma) \\
x_{e} , z_{S} &\geq 0 & & \forall e,S\\
\end{align*}
\end{minipage}
\quad
\begin{minipage}[t]{0.5\textwidth}
\begin{align*}
\min \sum_{i} \alpha_{i} &+ \gamma \\
\sum_{i} b_{i,e}  \cdot \alpha_{i} &\geq \beta_{e}  & &  \forall e & (x_e)\\
\gamma + \sum_{e \in S} \beta_{e} &\geq f(\vect{1}_{S})  & & \forall S & (z_S) \\
\alpha_{i} &\geq 0 & & \forall i 
\end{align*}
\end{minipage}

In the primal, $(\alpha_i)$ are the packing constraints of the given problem.  Constraints $(\beta_e)$ force the aforementioned  relation between $\vect x$ and $\vect z$.  Constraint $(\gamma)$ ensures that $\vect z$ represents a distribution.
Note that the primal constraints $(\gamma)$ and $(\beta_e)$, imply the box constraints $x_{e}  \leq 1 ~\forall e$. 
 
\paragraph{Algorithm.} 
Assume that function $F(\cdot)$ is $(\lambda, \mu)$-locally smooth.
Let $d$ be the maximal number of positive entries in a row, i.e., $d = \max_{i} |\{b_{ie}: b_{ie} > 0\}|$.  
Let $\rho$ the maximum divergence between positive row coefficients, i.e.\ $\rho = \max_{i} \max_{e,e': b_{ie' > 0}} b_{ie} / b_{ie'}$.   
The algorithm is given a prediction $\vect{x}^\pred$, i.e.\ with every arriving element $e$, it receives the values $x^\pred_{ie}$ for each resource $i$.  It uses this prediction to specify coefficients $\overline{\vect b}$, which are scaled from $\vect b$ in a specific manner.  The maximum divergence $\overline \rho$ is defined similarly as $\rho$ with $\overline{\vect b}$ replacing $\vect b$, i.e., $\overline{\rho} = \max_{i} \max_{e,e': \overline{b}_{ie' > 0}} \overline{b}_{i,e} / \overline{b}_{ie'}$.

The algorithm maintains a primal solution $\vect y\in[0,1]^{\cal E}$, computed with the primal-dual method with respect to the coefficients $\overline {\vect b}$. Its decision for element $e$ either follows the predicted solution $x^\pred_e$, or it follows $y_{e}$ in case infeasibility of the predicted solution $\vect x^\pred$ has been detected.

The value of $\overline{b}_{i,e}$ depends on coefficient $b_{i,e}$ and 
the prediction $x^\pred_e$. Specifically, $\overline{b}_{i,e} = b_{i,e}$ if
$x^{\pred}_{e} = 1$ and the predictive solution is \emph{not} feasible; $\overline{b}_{i,e} = \frac{1}{\eta} b_{i,e}$ otherwise. In both cases, packing constraints using $\overline{\vect b}$ are stronger than they would be with coefficients $\vect b$.

Intuitively, if we do not trust the prediction at all, i.e.\ $\eta=1$, then 
$\overline{b}_{i,e} = b_{i,e}$ and therefore $x_{e}, y_{e}$ would get a value proportional to the one returned by a primal-dual algorithm
in the classic setting.
Inversely, if we trust the prediction (i.e., $\eta$ is close to 0) and the prediction is feasible, then 
$\overline{b}_{i,e} = b_{i,e}$ when $x^{\pred}_{e} = 1$ and $\overline{b}_{i,e} = b_{i,e}/\eta$ when 
$x^{\pred}_{e} = 0$. Therefore, the modified constraint
$\sum_{e'} \overline{b}_{i,e'} y_{e'} \leq 1$ will likely prevent $y_{e}$, for $e$ such that $x^{\pred}_{e} = 0$, 
from getting a large value. Hence, $y_{e}$ for $e$ such that $x^{\pred}_{e} = 1$ could get a large value.
In the end of each iteration, we set the output solution $x_{e}$ roughly by scaling a factor $\frac{1}{1+\eta}$ to  $y_{e}$ or $x^{\pred}_{e}$
(depending on cases)
in order to maintain the feasibility and the consistency to the prediction.

The construction of $\vect{y}$ follows the scheme in \cite{Thang20:Online-Primal-Dual}. 
We recall the definition of the divergence factor $\overline{\rho} = \max_{i} \max_{e,e': \overline{b}_{ie' > 0}} \overline{b}_{i,e} / \overline{b}_{ie'}$.
So in particular, $\overline{\rho} \leq \rho/\eta$. 
Recall that the gradient in direction $e$ is $\nabla_{e} F(\vect{y}) = \partial F(\vect{y})/\partial y_{e}$.
By convention, when an element $e$ is not released, $\nabla_{e} F(\vect{y}) = 0$.  
While $\nabla_{e} F(\vect{y}) > 0$ --- i.e., increasing $y_{e}$ improves the objective value ---
and $\sum_{i} \overline{b}_{i,e}  \alpha_{e} \leq \frac{1}{\lambda} \nabla_{e} F(\vect{y})$, the primal variable $y_{e} $
and dual variables $\alpha_{i}$'s are increased by appropriate rates. We will argue in the analysis that 
the primal and dual solutions returned by the algorithm are feasible. 

Recall that by definition of the multilinear extension, 
%\marginpar{Better avoid game theory notation}
$$
\nabla_{e} F(\vect{y})
= F((\vect{y}_{-e}, 1)) - F((\vect{y}_{-e},0))
= \mathbb{E}_{R} \bigl[ f\bigl(\vect{1}_{R \cup \{e\}}\bigr) - f\bigl(\vect{1}_{R}\bigr) \bigr]
$$
where $(\vect{y}_{-e}, 1)$ denotes a vector which is identical to $\vect{y}$ on every coordinate different to $e$ and 
the value at coordinate $e$ is 1.  The vector $(\vect{y}_{-e}, 0)$ is defined similarly. 
The expectation is taken over random subset $R \subseteq \mathcal{E} \setminus \{e\}$ such that $e'$ is included with probability $y_{e'}$.
Therefore, during the iteration of the while loop with respect to element $e$, only $y_{e} $ is modified and $y_{e'} $ remains fixed for all $e' \neq e$, 
as a consequence $\nabla_{e} F(\vect{y})$ is constant during the iteration.
Moreover,  for every $e$, $F(\vect{y})$ and $\nabla_{e} F(\vect{y})$ 
can be efficiently approximated up to any required precision \cite{Vondrak10:Polyhedral-techniques}.

One important aspect of the primal-dual algorithm presented below, is that it works only with primal variables $\vect x$ and dual variables $\alpha$, and uses the multilinear extension $F$ instead of $f$.  However for the analysis, we will later show how to extend these solutions with variables $\vect z$ and $\vect \beta$, $\gamma$ both to show feasibility of the constructed solution and to analyze consistency and robustness of the algorithm.  

\begin{algorithm}[ht]
\begin{algorithmic}[1]  
\STATE All primal and dual variables are initially set to 0. 
\STATE Let $\vect{y}\in[0,1]^{\cal E}$ be such that $y_{e} = 0 ~\forall e$. 
\STATE At every step, always maintain $z_{S} = \prod_{e \in S} x_{e}  \prod_{e \notin S} (1 - x_{e} )$.
\FOR{each arrival of a new element $e$} 
	\IF{$x^{\pred}_{e} = 1$ \OR the predictive solution $\vect x^\pred$ is infeasible} 
	\STATE{ set $\overline{b}_{i,e} = b_{i,e}$} 
	\ELSE \STATE{ set $\overline{b}_{i,e} = b_{i,e}/\eta$}
	\ENDIF
	\WHILE{$\sum_{i} \overline{b}_{i,e}  \alpha_{i} \leq \frac{1}{\lambda} \nabla_{e} F(\vect{y})$ \AND $\nabla_{e} F(\vect{y}) > 0$}
%		\STATE During the while loop, always maintain $\beta_{e} \gets \frac{1}{\lambda} \nabla_{e} F(\vect{y})$. The value 
%			of $\beta_{e}$ will be fixed after the for loop related to resource $e$. 
		\STATE Some of the primal and dual variables are increased continuously as follows, where $\tau$ is the time during this process.
		% Let $\tau$ be the time in the execution of the algorithm. The dual variables evolve during the execution of the algorithm as follows. 
		% \marginpar{But $y_e$ is primal}
		\STATE Increase $y_{e} $ at a rate such that $ \frac{dy_{e}}{d\tau} \gets \frac{1}{\nabla_{e} F(\vect{y}) \cdot \ln(1+ d\overline{\rho}  )}$.	
				\label{algo-packing:x}
% 		\FOR{$i$ such that $\overline{b}_{i,e}  > 0$}	
% 			\STATE Increase $\alpha_{i}$ at a rate such that
% %				\begin{align*}
% 					$
% 					\frac{d \alpha_{i}}{d \tau}	\gets \frac{\overline{b}_{i,e}  \cdot \alpha_i}{\nabla_{e} F(\vect{y})}  + \frac{1}{d \lambda}
% 					$
% %				\end{align*}
% 					\label{algo-packing:alpha}
% 		\ENDFOR
	\ENDWHILE 
	\IF{$x^{\pred}_{e} = 1$ \AND the predictive solution is still feasible} 
		\STATE set $x_{e}  \gets \frac{1}{1+\eta} x^{\pred}_{e} = \frac{1}{1+\eta}$ 
	\ELSE\STATE set $x_{e}  \gets  \frac{1}{1+\eta} y_{e}$
	\ENDIF
\ENDFOR
\end{algorithmic}
\caption{Algorithm for Packing Problem.}
	\label{algo:packing}
\end{algorithm}

Note that once the prediction becomes infeasible, the algorithm works with the given $b$ coefficients and outputs the solution computed by 
$y_e$ scaled by $1/(1+\eta)$.

\paragraph{Primal variables.}
The vector $\vect x$ is completed by $\vect z$ to form a complete solution to
the primal linear program, by setting $z_{S} = \prod_{e \in S} x_{e}  \prod_
{e \notin S} (1 - x_{e} )$.

\paragraph{Dual variables.} 
Variables $\alpha_{i}$'s  are constructed in the algorithm. 
%Let $\vect{x} = (x_{e'})_{e'}$ and 
%Let $\vect{y} = (y_{e'})_{e'}$ be the current vectors of the algorithm and 
%Let $\vect{y}^{e}$ be the vector $\vect{y}$ just after the while loop with respect to resource $e$.
Define $\gamma = \frac{\mu}{\lambda} F(\vect{y})$ and 
$\beta_{e} = \frac{1}{\lambda} \cdot \nabla_{e} F(\vect{y})$.
%During the while loop with respect to resource $e$, by the observation above (i.e., $\nabla_{e} F(\vect{y})$ is constant during the iteration), 
%we have $\beta_{e} = \frac{1}{\lambda} \cdot \nabla_{e} F(\vect{y})$. 

The following lemma provides a lower bound on the $\alpha$ variables.
%Remark that the monotonicity of the gradient is crucial in the analysis of \cite{AzarBuchbinder16:Online-Algorithms}, in particular 
%to prove the bounds on $x$-variables. However, by our approach the gradient monotonicity  
%is not needed. 

\begin{restatable}{lemma}{BoundAlpha}
\label{lem:bound-alpha}
At any moment during the iteration related to element $e$,  
for every constraint $i$
it always holds that 
%$$
%\alpha_{i}	\geq  \frac{\beta_{e}}{\max_{e'}  \overline{b}_{i,e'}  \cdot d} 
%		\left[ \exp\biggl( \ln \bigl(1+ d\overline{\rho} \bigr) 
%				\cdot \sum_{e'} \overline{b}_{i,e'}  \cdot y_{e'}  \biggr) - 1 \right].
\begin{align}
\alpha_{i}	&\geq  \frac{\nabla_{e} F(\vect{y})}{\max_{e'}  \overline{b}_{i,e'}  \cdot d \lambda} 
		\left[ \exp\biggl( \ln \bigl(1+ d\overline{\rho} \bigr) 
				\cdot \sum_{e'} \overline{b}_{i,e'}  \cdot y_{e'}  \biggr) - 1 \right]. \label{eq:inv_alpha}
\end{align}
\end{restatable}
\begin{proof}
Fix a constraint $i$. We prove the claimed inequality by induction.
In the very beginning of the algorithm, when no elements are released yet, the inequality holds since both sides are 0.  
Consider any moment $\tau$ during the loop corresponding to an arriving element $e$.
Assume that the inequality holds at the beginning of the iteration step. We will show that the inequality still holds after the step.

As $F$ is a multilinear extension, by its very definition, $F$ is linear in $y_{e}$. Hence 
$\nabla_{e} F(\vect{y})$ is independent of $y_{e} $. 
Moreover, during the loop corresponding to resource $e$, only $y_{e} $ is modified, leaving  $y_{e'} $ unchanged for all $e' \neq e$.  
As as a result, 
$d \nabla_{e} F(\vect{y})/d \tau = 0$. Therefore, the rate at time $\tau$ of the the right-hand-side of Inequality~\eqref{eq:inv_alpha} is:

\begin{align*}
&\frac{\nabla_{e} F(\vect{y})}{\max_{e'}  \overline{b}_{i,e'}  \cdot d \lambda } \cdot \ln \bigl(1+ d\overline{\rho} \bigr) \cdot \overline{b}_{i,e}  \cdot
		 \frac{d y_{e} }{d \tau} \cdot \exp \biggl( \ln \bigl(1+ d\overline{\rho} \bigr) \cdot \sum_{e'} \overline{b}_{i,e'}  \cdot y_{e'}  \biggr) \\
&\leq  
\frac{\nabla_{e} F(\vect{y})}{\max_{e'}  \overline{b}_{i,e'}  \cdot d \lambda} \cdot \ln \bigl(1+ d\overline{\rho} \bigr) \cdot \overline{b}_{i,e}  \cdot
		 \frac{1}{\nabla_{e} F(\vect{y}) \cdot \ln(1+ d\overline{\rho} )}
		 \cdot \biggl( \frac{\max_{e'} \overline{b}_{i,e'}  \cdot d\lambda \cdot \alpha_{i}}{ \nabla_{e} F(\vect{y}) } + 1 \biggr) \\
&\leq \frac{\overline{b}_{i,e}  \cdot \alpha_{i}}{\nabla_{e} F(\vect{y})}  + \frac{1}{d \lambda} = \frac{d\alpha_{i}}{d \tau},
\end{align*}
where in the first inequality we use the induction hypothesis and the increasing rate of $y_{e}$.
So at any time during the iteration related to $e$, the increasing rate of the left-hand side is always larger than that of the right-hand side. 
Hence, the lemma follows.
\end{proof}

\begin{lemma}		\label{lem:packing-primal-feasible}
The primal variables constructed by the algorithm are feasible. 
\end{lemma}
\begin{proof}
%We prove the primal feasibility. 
First observe that if during the execution of the algorithm in the iteration related to some element $e$,  we have
$\sum_{e'} \overline{b}_{i,e'}  y_{e'}  > 1$ for some constraint $i$ then by Lemma~\ref{lem:bound-alpha},
$$
\alpha_{i} 
>  \frac{\nabla_{e} F(\vect{y})}{\max_{e'}  \overline{b}_{i,e'}  \cdot d\lambda } 
		\left[ \exp\biggl( \ln \bigl(1+ d\overline{\rho} \bigr)  \biggr) - 1 \right]
= \frac{\overline{\rho} \cdot \nabla_{e} F(\vect{y})}{\lambda \max_{e'}  \overline{b}_{i,e'} }
\geq \frac{\nabla_{e} F(\vect{y})}{\lambda \overline{b}_{i,e} }
$$
Therefore, $\sum_{i} \overline{b}_{i,e}  \alpha_{i} > \frac{1}{\lambda} \nabla_{e} F(\vect{y})$ and hence the algorithm would have 
stopped increasing $y_{e} $ at some earlier point. 
Therefore, every constraint $\sum_{e'} \overline{b}_{i,e'}  y_{e'}  \leq 1$ is always maintained during the execution of the algorithm. 
%(Note that by definition $\overline{b}_{i,e'} \geq b_{i,e'} \geq 0$ thus it also holds that $\sum_{i} b_{i,e'}  y_{e'}  \leq 1$.)

Secondly, we show primal feasibility (even when the prediction oracle provides an infeasible solution).
If the prediction oracle provides a feasible solution, we set $S_1=\{e: x^\pred_e=1\}$ and $S_2=\emptyset$.
Otherwise, let $e^{*}$ be the the first element for which the prediction oracle provides an infeasible solution. 
Let $S_{1}$ be the set of all resources $e$ such that $x^{\pred}_{e} = 1$ and $e$ is released before $e^{*}$.
Further let
$S_{2} = \{e: x^{\pred}_{e} = 1\} \setminus S_{1}$.
For every constraint $i$, 
\begin{align*}
\sum_{e} b_{i,e} x_{e} 
&= \sum_{e: x^{\pred}_{e} = 1} b_{i,e} x_{e} +  \sum_{e: x^{\pred}_{e} = 0} b_{i,e} x_{e} \\
&= \sum_{e \in S_{1}}  b_{i,e} x_{e} +  \sum_{e \in S_{2}} b_{i,e} x_{e}  +   \sum_{e: x^{\pred}_{e} = 0} b_{i,e} x_{e} \\
&\leq \frac{1}{1+\eta} \cdot 1
+ \eta \cdot \sum_{e \in S_{2}} \overline{b}_{i,e} x_{e} 
+ \eta \cdot \sum_{e: x^{\pred}_{e} = 0} \overline{b}_{i,e} x_{e} \\
&= \frac{1}{1+\eta}  
+ \frac{\eta}{1+\eta}  \sum_{e \in S_{2}} \overline{b}_{i,e} y_{e} 
+ \frac{\eta}{1+\eta}  \sum_{e: x^{\pred}_{e} = 0} \overline{b}_{i,e} y_{e} \\
&\leq \frac{1}{1+\eta}
+ \frac{\eta}{1+\eta}  \sum_{e} \overline{b}_{i,e} y_{e} \\
&\leq \frac{1}{1+\eta}
+ \frac{\eta}{1+\eta}  \cdot 1 = 1.
\end{align*}
The first inequality is due to: (1) the feasibility of the prediction restricted on $S_{1}$, i.e., 
$\sum_{e \in S_{1}} b_{i,e} x^{\pred}_{e} \leq 1$; and (2) 
$x_{e} = \frac{1}{1 + \eta} = \frac{1}{1 + \eta} x^{\pred}_{e}$
for $e \in S_{1}$; and (3) the definitions of $\overline{b}_{i,e}$ in Algorithm~\ref{algo:packing}.
The third equality follows by the algorithm: 
$x_{e} = \frac{1}{1 + \eta} \cdot y_{e}$
for $e \notin S_{1}$. 
The last inequality holds since $\sum_{e} \overline{b}_{i,e} y_{e} \leq 1$ by the observation made at the beginning of the lemma.
Hence, $\sum_{e} b_{i,e} x_{e} \leq 1$ for every constraint $i$.

Besides, by definition of $z_{S} = \prod_{e \in S} x_{e} \prod_{e \notin S} (1 - x_{e})$
where $0 \leq x_{e} \leq 1$ for all $e$, the identity $\sum_{S} z_{S} = 1$ always holds. 
In fact, if one chooses an element $e$ with probability $x_{e}$ then $z_{S}$ is the 
probability that the set of selected elements is $S$. So the total probability $\sum_{S} z_{S}$ must be 1.
Similarly, 
$\sum_{S: e \in S} z_{S} = x_{e} \sum_{S' \subset E \setminus \{e\}} \prod_{e' \in S'} x_{e'} \prod_{e' \notin S'} (1 - x_{e'}) = x_{e}$
since $\sum_{S' \subset E \setminus \{e\}} \prod_{e' \in S'} x_{e'} \prod_{e' \notin S'} (1 - x_{e'}) = 1$ (by the same argument). 

Therefore, the solution $(\vect{x}, \vect{z})$ is primal feasible.
\end{proof}

\begin{lemma}
The dual variables defined as above are feasible. 
\end{lemma}
\begin{proof}
The first dual constraint $\sum_{i} b_{i,e} \alpha_{i} \geq \beta_{e}$ is satisfied by the while loop condition of the algorithm
and the definition of $\beta_{e}$.
The second dual constraint $\gamma + \sum_{e \in S} \beta_{e} \geq f(\vect{1}_{S})$ reads
\begin{align*}
	\frac{1}{\lambda} \sum_{e \in S} \nabla_{e} F(\vect{y}) + \frac{\mu}{\lambda} F(\vect{y}) \geq F(\vect{1}_{S}), 
\end{align*}
which is, by arranging terms, exactly the $(\lambda, \mu)$-max-local smoothness of $F$. 
(Recall that $F(\vect{1}_{S}) = f(\vect{1}_{S})$.)
Hence, the lemma follows.
\end{proof}

%We are now ready to prove the main theorem. 

%%
%\begin{theorem}	\label{thm:packing}
%Assume that the multilinear extension is $(\lambda, \mu)$-max-locally-smooth.
%Then, the algorithm is $\Omega\bigl( \frac{\lambda}{2\ln(1+ d\rho/\eta ) + \mu} \bigr)$-competitive
%where ...
%\end{theorem}

\Packing*
\begin{proof}
\paragraph{Robustness.}
First, we bound the increases of $F(\vect{y})$ and of the dual objective value --- which we denote by $D$ --- at any time $\tau$ in the execution of 
Algorithm~\ref{algo:packing}.
The derivative of $F(\vect{y})$ with respect to $\tau$ is:
\begin{align*} %	\label{eq:packing-primal}
\nabla_{e} F(\vect{y}) \cdot \frac{d y_{e} }{ d \tau}
= \nabla_{e} F(\vect{y}) \cdot\frac{1}{\nabla_{e} F(\vect{y}) \cdot \ln(1+ d\overline{\rho})}
= \frac{1}{\ln(1+ d\overline{\rho})}
\end{align*}
Besides, the rate of the dual at time $\tau$ is:
\begin{align*}
\frac{d D}{d \tau} 
&=  \sum_{i} \frac{d \alpha_{i}}{d \tau} + \frac{d \gamma}{d \tau}  
= \sum_{i: \overline{b}_{i,e}  > 0} \biggl( \frac{\overline{b}_{i,e}  \cdot \alpha_{i}}{\nabla_{e} F(\vect{y})}  + \frac{1}{d\lambda} \biggr) 
		+ \frac{\mu}{\lambda} \frac{d F(\vect{y})}{d \tau}  \\
&= \sum_{i: \overline{b}_{i,e}  > 0} \frac{\overline{b}_{i,e}  \cdot \alpha_{i}}{\nabla_{e} F(\vect{y})}  +  \sum_{i: \overline{b}_{i,e}  > 0} \frac{1}{d\lambda}
		+ \frac{\mu}{\lambda} \cdot \frac{1}{\ln(1+ d\overline{\rho})} \\
&\leq \frac{2}{\lambda} + \frac{\mu}{\lambda \cdot \ln(1+ d\overline{\rho})}
=  \frac{2\ln(1+ d\overline{\rho}) + \mu}{\lambda \cdot \ln(1+ d\overline{\rho})},
\end{align*}
where the inequality holds since during the algorithm
$\sum_{i} \overline{b}_{i,e}  \cdot \alpha_{i} \leq  \frac{1}{\lambda} \nabla_{e} F(\vect{y})$.
Hence, the ratio between $F(\vect{y})$ and the dual $D$ is at least $\frac{\lambda}{2\ln(1+ d\overline{\rho}) + \mu}$.

Besides, $x_{e} = \frac{1}{1 + \eta} \geq \frac{1}{1 + \eta} y_{e}$ if $x^{\pred}_{e} = 1$ and the predictive solution is still feasible; and
 $x_{e} = \frac{1}{1 + \eta} y_{e}$ otherwise. Therefore, 
 $\vect{x} \geq \frac{\vect{y}}{1 + \eta} $ and so 
 $F(\vect{x}) \geq F\bigl(\frac{\vect{y}}{1 + \eta}\bigr)$ by 
 monotonicity\footnote{Note that this is the only step in the analysis we use the monotonicity of $f$, which implies the monotonicity of $F$.} 
 of $F$. 
 Hence the robustness is at least 
$$
\frac{F(\vect{x})}{F(\vect{y})} \cdot \frac{\lambda}{2\ln(1+ d\overline{\rho}) + \mu}
\geq \min_{\vect{0} \leq \vect{u} \leq \vect{1}} \frac{F(\frac{1}{1 + \eta} \vect{u})}{F(\vect{u})} \cdot \frac{\lambda}{2\ln(1+ d\rho/\eta ) + \mu} 
$$ 
where the latter is due to $\overline{\rho} \leq \rho/\eta$.

\paragraph{Consistency.} By our algorithm, for every element $e$, if $x^{\pred}_{e} = 1$ (and the prediction 
oracle provides a feasible solution) 
then $x_{e} = \frac{1}{1+\eta}$. Hence, the consistency of the algorithm 
$F(\vect{x})/F(\vect{x}^{\pred}) \geq F(\frac{\vect{x}^{\pred}}{1 + \eta})/F(\vect{x}^{\pred}) \geq r(\eta)$.
\end{proof}

%Note that the competitive ratio is 
%the same up to a constant factor as the performance guarantee for maximizing a linear function
%under packing constraints. Specifically, if function $f$ is linear then the smooth parameters are
%$\lambda = \mu = 1$.
%
%\paragraph{Remark.} One can define $\overline{b}_{e} = b_{e}/(\frac{1}{1+\eta})$ instead of 
%$\overline{b}_{e} = b_{e} (1 + \eta)$ but be careful when $\eta = 1$. One can assume that 
%$\eta \in (0,1)$.
\subsection{Applications}

\subsubsection{Applications to linear functions}
When the objective $f$ can be expressed as a monotone linear functions, its multilinear extension $F$ 
is $(1,1)$-locally-smooth. Moreover, $r(\eta) = 1/(1+\eta)$. Consequently, 
Algorithm \ref{algo:packing} provides a $1/(1 + \eta)$-consistent
and $O\bigl(1/\ln(1+ d \rho/\eta)\bigr)$-robust fractional solution for the
online linear packing problem.

\subsubsection{Applications to online submodular maximization}	\label{sec:sub-max}
Consider the online problem of maximizing a monotone submodular function subject to packing constraints. 
A set-function $f: 2^{\mathcal{E}} \rightarrow \mathbb{R}+$ is \emph{submodular} if
$f(S \cup e) - f(S) \geq f(T \cup e) - f(T)$ for all $S \subset T \subseteq \mathcal{E}$. 
Let $F$ be the multilinear extension of a monotone submodular function $f$. Function $F$
admits several useful properties: (i) if $f$ is monotone then so is $F$; (ii) $F$ is concave in any
positive direction, i.e., $\nabla F(\vect{x}) \geq \nabla F(\vect{y})$ for all $\vect{x} \leq \vect{y}$
($\vect{x} \leq \vect{y}$ means $x_{e} \leq y_{e} ~\forall e$). 

\begin{lemma}	\label{lem:sub-max-locally-smooth}
Let $f$ be an arbitrary monotone submodular function. Then, its multilinear extension 
$F$ is (1,1)-locally-smooth.
\end{lemma}
\begin{proof}
As $F$ is the linear extension of a submodular function, 
$\nabla_{e} F(\vect{x}) = \mathbb{E}_{R}\bigl[ f\bigl(\vect{1}_{R \cup \{e\}}\bigr) - f\bigl(\vect{1}_{R}\bigr) \bigr]$
where $R$ is a random subset of $\mathcal{E} \setminus \{e\}$ such that $e'$ is included with probability $x_{e'} $.
For any subset $S = \{e_{1}, \ldots, e_{\ell}\}$, we have
\begin{align*}
F(\vect{x}) + \sum_{e \in S} \nabla_{e} F(\vect{x}) 
&= F(\vect{x}) + \sum_{e \in S} \mathbb{E}_{R} \bigl[ f\bigl(\vect{1}_{R \cup \{e\}}\bigr) - f\bigl(\vect{1}_{R}\bigr) \bigr] \\
&= \mathbb{E}_{R} \biggl[ f(\vect{1}_{R}) + \sum_{e \in S} \bigl[ f\bigl(\vect{1}_{R \cup \{e\}}\bigr) - f\bigl(\vect{1}_{R}\bigr) \bigr] \biggl] \\
&\geq \mathbb{E}_{R} \biggl[ f(\vect{1}_{R}) + \sum_{i=1}^{\ell} \bigl[ f(\vect{1}_{R \cup \{e_{1}, \ldots, e_{i}\}}) - 
								f(\vect{1}_{R \cup \{e_{1}, \ldots, e_{i-1}\}}) \bigr] \biggr] \\
&= \mathbb{E}_{R} \bigl[ f(\vect{1}_{R \cup S})  \bigl]  \geq \mathbb{E}_{R} \bigl[ f(\vect{1}_{S})  \bigl] \\
&= F\bigl( \vect{1}_{S} \bigr)
\end{align*}
the first inequality is due to the submodularity $f$ and the second one due to its monotonicity.
The lemma follows.
\end{proof}

%\begin{lemma}
%Algorithm \ref{algo:packing} yields a  $O(\ln r(\mathcal{M}))$-competitive fractional solution where $r$ is the rank of 
%matroid $\mathcal{M}$.
%\end{lemma}

%The previous lemma and Theorem \ref{thm:packing} lead to the following result. 

\begin{proposition}	\label{prop:max-submodular}
For any $0 < \eta \leq 1$, Algorithm \ref{algo:packing} gives a $(1 - \eta)$-consistent
and $O\bigl(1/\ln(1+ d \rho/\eta)\bigr)$-robust fractional solution to
the problem of online submodular maximization under packing constraints.
\end{proposition}
\begin{proof}
We first bound $r(\eta) = \min_{\vect{0} \leq \vect{u} \leq \vect{1}} F(\frac{1}{1 + \eta} \vect{u})/F(\vect{u})$. 
By the non-negativity and the concavity in positive direction of $F$, for any set $S \subseteq \mathcal{E}$, 
\begin{align}
F\biggl(\frac{\vect{u}}{1 + \eta}  \biggr) 
&\geq F(\vect{u}) - \biggl \langle \nabla F\biggl(\frac{\vect{u}}{1 + \eta} \biggr), \vect{u} - \frac{\vect{u}}{1 + \eta} \biggr \rangle
\label{eq:max-sub-1}
\\
F\biggl(\frac{\vect{u}}{1 + \eta} \biggr) 
 &\geq F(\vect{0}) + \biggl \langle \nabla F\biggl(\frac{\vect{u}}{1 + \eta} \biggr), \frac{\vect{u}}{1 + \eta}\biggr \rangle
 \geq  \biggl \langle \nabla F\biggl(\frac{\vect{u}}{1 + \eta} \biggr), \frac{\vect{u}}{1 + \eta} \biggr \rangle.
\label{eq:max-sub-2}
\end{align}
Therefore,
\begin{align}
\frac{F\biggl(\frac{\vect{u}}{1 + \eta}\biggr)}{F(\vect{u})}
&\geq 1 - \frac{\eta}{1 + \eta} \cdot \frac{\biggl \langle \nabla F\biggl(\frac{\vect{u}}{1 + \eta}\biggr), \vect{u} \biggr \rangle}{F(\vect{u})} 	\tag{by (\ref{eq:max-sub-1})}\\
&\geq 
1 - \frac{\eta}{1 + \eta} \cdot \frac{\biggl \langle \nabla F\biggl(\frac{\vect{u}}{1 + \eta}\biggr), \vect{u} \biggr \rangle}{F\biggl(\frac{\vect{u}}{1 + \eta}\biggr)}  	\tag{by monotonicity of $F$}\\
&\geq 
1 - \frac{\eta}{1 + \eta} \cdot \frac{\biggl \langle \nabla F\biggl( \frac{\vect{u}}{1 + \eta} \biggr), \vect{u} \biggr \rangle}{\biggl \langle \nabla F\biggl(\frac{\vect{u}}{1 + \eta}\biggr), \frac{\vect{u}}{1 + \eta}\biggr \rangle}
	\tag{by (\ref{eq:max-sub-2})} \\
&= 1 - \eta. 	\notag
\end{align}
So, $r(\eta) \geq 1 - \eta$. Therefore, the proposition holds by Theorem~\ref{thm:packing} and by the (1,1)-locally-smoothness of $F$ (Lemma~\ref{lem:sub-max-locally-smooth}).
\end{proof}

One can derive online randomized algorithms for the integral variants of these problems by rounding the fractional solutions.
For example, using the online contention resolution rounding schemes \cite{FeldmanSvensson16:Online-contention}, 
one can obtain randomized algorithms for several specific constraint polytopes, for example, knapsack polytopes, 
matching polytopes and matroid polytopes.

%!TEX root = main.tex

\section{Ad-Auction Revenue Maximization}		\label{sec:ad-auction}

\paragraph{Problem.} In the problem, we are given $m$ buyers, each buyer $1 \leq i \leq m$ has a budget $B_{i}$.  
Items arrive online and at the arrival of item $e$, 
each buyer $1 \leq i \leq n$ provides a bid $b_{i,e} \ll B_{i}$ for buying $e$.
In the integral variant of the problem, one needs to allocate the item to at most one of the buyers. For every buyer, the total bids of the allocated items should not exceed the budget. Note that this forms a packing constraint, similar to the ones studied in the previous section, except that we don't normalize the constraints by the factor $1/B_i$.
In the fractional variant of the problem, items can be allocated in fractions to buyers, not exceeding a total of $1$ among the fractions.  In both variants, the objective is to maximize the total revenue received from all buyers, which is the sum over all  allocations of corresponding bid, possibly multiplied with the fraction of the allocation.

In out setting, every item $e$ arrives with a prediction, suggesting a buyer to whom to allocate that item, if any.
For convenience we denote the items by the integers from $1$ to $n$.

\subsection{The algorithm}

\paragraph{Formulation.}
The fractional variant of problem can be expressed as the following linear program, where $x_{i,e}$ indicates the fraction at which item $e$ is allocated to buyer $i$.

\begin{minipage}[t]{0.45\textwidth}
\begin{align*}
\max  \sum_{e=1}^{n} & \sum_{i=1}^{m} b_{i,e} x_{i,e} & \\
\sum_{i=1}^{m} x_{i,e}  &\leq 1 & &  \forall e & (\beta_e) \\
\sum_{e=1}^{n} b_{i,e} x_{i,e} &\leq  B_{i} & & \forall i & (\alpha_i) \\
x_{i,e} &\geq 0 & & \forall i, e\\
\end{align*}
\end{minipage}
\quad
\begin{minipage}[t]{0.5\textwidth}
\begin{align*}
\min \sum_{i=1}^{n} B_{i} \alpha_{i} &+ \sum_{e=1}^{m} \beta_{e} \\ 
b_{i,e} \alpha_{i} + \beta_{e} &\geq b_{i,e}  & &  \forall i, e & (x_{i,e})\\  \\
\alpha_{i}, \beta_{e} &\geq 0 & & \forall i,e \\
\end{align*}
\end{minipage}
In the primal linear program, the first constraint ensures that an item can be allocated to at most one buyer and 
the second constraint ensures that a buyer $i$ does not spend more than the budget $B_{i}$.
The dual of the relaxation is presented in the right.  

\paragraph{Algorithm.} 

For the sake of simplicity in the algorithm description, we introduce a fictitious buyer, denoted as $0$, such that  
$b_{0, e} = 0$ for all items $e$. The non-assignment of an item $e$ to any buyer can be seen as being assigned 
to the fictitious buyer $0$ with revenue 0.
The purpose of buyer $0$   is to simplify the description of the algorithm.
 We define the constant $C = (1 + R_{\max})^{\frac{\eta}{R_{\max}}}$ where 
$R_{\max} = \max_{i,e} \frac{b_{i,e}}{B_{i}}$. Adapting the primal-dual algorithm presented in the previous section, we obtain the following algorithm for the fractional ad-auction problem.

\begin{algorithm}[ht]
\begin{algorithmic}[1]  
\STATE All primal and dual variables are initially set to 0.
\STATE For the analysis, we maintain for every buyer $i$ two sets $N(i), M(i)$, both initially empty. 
\FOR{each arrival of a new item $e$}
	\STATE Let $i^{*}$ be the buyer such that $x^{\pred}_{i^{*},e} = 1$. If either the prediction is not feasible or 
	 	there is no such $i^{*}$ (i.e., $x^{\pred}_{i',e} = 0 ~\forall i'$) then
		set $i^{*} = 0$.  
	\STATE Let $i$ be the buyer that maximizes $b_{i',e} (1 - \alpha_{i'})$. 
		If $b_{i,e} (1 - \alpha_{i}) \leq 0$ then 
		set $i = 0$.
	\STATE Set $\beta_{e} \gets \max \bigl \{0,  b_{i,e} (1 - \alpha_{i}) \bigr \}$	 \COMMENT{for the purpose of analysis}
	\IF{$b_{i,e} < b_{i^{*},e}$}
		\STATE Set $x_{i,e} \gets \eta$ and $x_{i^{*},e} \gets 1-\eta$.
		\STATE Define $\overline{b}_{i,e} = b_{i,e}/\eta$ 
		\STATE $N(i^{*}) \gets N(i^{*}) \cup \{e\}$. 		\COMMENT{for the purpose of analysis}
	\ELSE 
		\STATE Set $x_{i,e} \gets 1$ and define $\overline{b}_{i,e} = b_{i,e}$
		\COMMENT{includes case $\vect x^\pred$ is infeasible}
	\ENDIF
	\STATE Update $M(i) \gets M(i) \cup \{e\}$.  \COMMENT{for the purpose of analysis}
	\STATE Update $\alpha_{i} \gets \alpha_{i}\left( 1 + \frac{b_{i,e}}{B_{i}} \right) 
										+  \frac{b_{i,e}}{B_{i}} \cdot \frac{1}{C - 1}$.	
\ENDFOR
% \STATE Scale $\vect x$ by factor $1-R_\max$
\end{algorithmic}
\caption{Algorithm for Ad-Auctions Revenue Maximization.}
\label{algo:ad-auctions}
\end{algorithm}

\begin{lemma}	\label{lem:alpha}
For every $i$, it always holds that 
$$
\alpha_{i} \geq \frac{1}{C - 1} \biggl( C^{\frac{\sum_{e \in M(i)} b_{i,e}}{\eta B_{i}}} - 1 \biggr).
$$
\end{lemma}
\begin{proof}
We adapt the proof in \cite{BuchbinderNaor09:The-Design-of-Competitive} 
with a slight modification. 
The primal inequality $\alpha_i$ is proved by induction on the number of released items. 
Initially, when no item is released, the inequality  is trivially true. 
Assume that the inequality holds right before the arrival of an item $e$. 
The inequality remains unchanged for all but the buyer $i$ selected in line 5 of the algorithm, i.e.\ $i$ maximizes $b_{i,e} (1 - \alpha_{i})$. 
We denote by $\alpha_i$  the value before the update triggered by the arrival of $e$ and $\alpha'_i$ its value after the update.
We have
\begin{align*}
\alpha'_{i}
&= \alpha_{i} \cdot \left( 1 + \frac{b_{i,e}}{B_{i}} \right) 
										+  \frac{b_{i,e}}{B_{i}} \cdot \frac{1}{C - 1}	\\
&\geq \frac{1}{C - 1} \biggl( C^{\frac{\sum_{e' \in M(i) \setminus e } b_{i,e}}{\eta B_{i}}} - 1 \biggr)
			\cdot \left( 1 + \frac{b_{i,e}}{B_{i}} \right) 
										+   \frac{b_{i,e}}{B_{i}} \cdot \frac{1}{C - 1} \\
&=  \frac{1}{C - 1} \biggl( C^{\frac{\sum_{e' \in M(i) \setminus e } b_{i,e}}{\eta B_{i}}} 
						\cdot \left( 1 + \frac{b_{i,e}}{B_{i}} \right)  - 1 \biggr) 	\\						
&\geq  \frac{1}{C - 1} \biggl( C^{\frac{\sum_{e' \in M(i) \setminus e } b_{i,e}}{\eta B_{i}}} 
						\cdot  C^{\frac{b_{i,e} }{\eta B_{i}}}   - 1 \biggr) \\
&= \frac{1}{C - 1} \biggl( C^{\frac{\sum_{e \in M(i)} b_{i,e}}{\eta B_{i}}} - 1 \biggr).
\end{align*}
The first inequality is due to the induction hypothesis.
The second inequality is due to this sequence of transformations. 
For any $0 < y \leq z \leq 1$ we have
\begin{align*}
\frac{\ln(1+y)}{y} &\geq \frac{\ln(1+z)}{z} 
\\
\Leftrightarrow \quad \ln(1+y) &\geq \ln(1+z)y / z 
\\
\Leftrightarrow \quad 1+y &\geq (1+z)^{y/z},
\end{align*}
which we apply with $y=b_{i,e}/{B_i}$ and $z=R_{\max}$. By the definition of $C$ we obtain
$$
\biggl( 1+R_{\max} \biggr)^{\frac{1}{R_{\max}} \cdot \frac{b_{i,e} }{B_{i}}} 
= C^{\frac{b_{i,e}}{\eta B_{i}}}.
$$
This completes the induction step.
\end{proof}

\begin{lemma}	\label{lem:ad-auctions-primal-feasibility}
The primal solution is feasible up to a factor $(1 + R_{\max})$.
\end{lemma}
\begin{proof}
The first primal constraint $\sum_{e} x_{i,e} \leq 1$ follows the values of $x_{i,e}$ and $x_{i^{*},e}$ in Algorithm~\ref{algo:ad-auctions}.
For the second primal constraint, we first prove that $\sum_{e=1}^{m} b_{i,e} \leq  B_{i}$ for every $i$. 
By Lemma~\ref{lem:alpha}, for every $i$, 
$$
\alpha_{i} \geq \frac{1}{C - 1} \biggl( C^{\frac{\sum_{e \in M(i)} b_{i,e}}{\eta B_{i}}} - 1 \biggr)
$$
So whenever $\sum_{e \in M(i)} b_{i,e} \geq \eta B_{i}$, we have $\alpha_{i} \geq 1$; so the algorithm 
stops allocating items to buyer $i$. Hence, the buyer $i$ can be allocated at most one additional item once her budget is 
already saturated. Hence, $\sum_{e \in M(i)} b_{i,e}  < \eta B_{i} + \max_{e} b_{i,e}$. 
Therefore, 
\begin{align*}
\sum_{e=1}^{m} b_{i,e}x_{i,e} = 
\sum_{e \in M(i)} b_{i,e}x_{i,e} + \sum_{e \in N(i)} b_{i,e}x_{i,e}
<  B_{i} + \max_{e} b_{i,e}
\end{align*}
where the latter is due to the feasibility of $N(i)$, the set of items assigned by the prediction to buyer $i$,
i.e., $\sum_{e \in N(i)} b_{i,e}x_{i,e} \leq  (1 - \eta) \sum_{e} b_{i,e} x^{\pred}_{i,e} \leq (1 - \eta) B_{i}$.
So, $\sum_{e=1}^{m} b_{i,e}x_{i,e} \leq B_{i}(1 + R_{\max})$.
\end{proof}

\begin{theorem}
The algorithm is $(1 - \eta)$-consistent and $\frac{1 - 1/C}{1 + R_{\max}}$-robust.
The robustness tends to $1 - e^{-\eta}$ when $R_{\max}$ tends to 0.
\end{theorem}
\begin{proof}
First, we establish robustness. At the arrival of item $e$, the increase in the primal is 
$$
\begin{cases}
	(1-\eta) \cdot b_{i^{*},e} + \eta \cdot b_{i,e} & \text{ if } b_{i,e} < b_{i^{*},e}, \\
	b_{i,e} & \text{ if } b_{i,e} \geq b_{i^{*},e} 
\end{cases}
$$
which is always larger than $b_{i,e}$. Besides, the increase in the dual is 
%\begin{align*}
%B_{i} \Delta \alpha_{i} + B_{i^{*}} \Delta \alpha_{i^{*}} + \beta_{e}
%&= \frac{\eta}{1+\eta} \biggl( b_{i,e} \alpha_{i} + \frac{b_{i,e}}{c-1} \biggr)
%+ \frac{1}{1+\eta} \biggl( b_{i^{*},e} \alpha_{i^{*}} + \frac{b_{i^{*},e}}{c-1} \biggr) \\
%%
%& \qquad + \frac{\eta}{1+\eta} b_{i,e} (1 - \alpha_{i}) + \frac{1}{1+\eta} b_{i^{*},e} (1 - \alpha_{i^{*}}) \\
%%
%&= \biggl( \frac{1}{1+\eta} \cdot b_{i^{*},e} +  \frac{\eta}{1+\eta} \cdot b_{i,e} \biggr) \biggl( 1 + \frac{1}{c-1} \biggr)
%\end{align*}
\begin{align*}
B_{i} \Delta \alpha_{i}  + \beta_{e}
&= b_{i,e} \alpha_{i} + \frac{b_{i,e}}{C-1} 
 + b_{i,e} (1 - \alpha_{i}) 
 = \biggl( 1 + \frac{1}{C-1} \biggr) b_{i,e}
 = \frac{C}{C - 1}.
 \end{align*}
Hence, by Lemma~\ref{lem:ad-auctions-primal-feasibility}, the robustness is 
$\frac{C-1}{C} \cdot \frac{1}{1+R_{\max}}$.

Secondly we establish consistency quite easily. At every time the prediction solution gets a profit $b_{i^{*},e}$, the algorithm achieves a profit 
at least $(1 - \eta) b_{i^{*},e}$.
\end{proof}

\subsection{Experiments}

We run an experimental evaluation of Algorithm~\ref{algo:ad-auctions}. Source code is publicly available\footnote{\url{http://www.lip6.fr/christoph.durr/packing}}.
For this purpose we constructed a random instance on 100 buyers and 10,000 items, adapting the model described in \cite{Lavastida20:Predictions-Matching}. Every item $e$ receives a bid from exactly 6 random buyers. This means that in average every bidder makes 600 bids. The value of the bid follows a lognormal distribution with mean $1/2$ and deviation $1/2$ as well. By choosing the budget of the bidders we can tune the hardness of the instance, under the constraint that $R_{\max}$ remains reasonably small. We choose the budget as a $1/10$ fraction of the total bids, leading to a value $R_{\max}\approx 0.19$.

The competitive ratio is determined using the fractional offline solution. The integrality gap of our instance is very close to $1$, roughly $1.0001$ in fact. For the prediction, we computed first the optimal integral solution, which is a partial mapping from items to buyers. This solution was perturbed as suggested by \cite{BamasMaggiori20:The-Primal-Dual-method}: Every item was mapped, with probability $\epsilon$,  to a uniformly chose random buyer, among the buyers who bid on the item.

The experiments are shown in Figure~\ref{fig:experiments}, and indicate clearly the benefit of the prediction when its perturbation is small.  Note that the ratio is not close to $1$ in that case for $\eta$ close to zero, because the algorithm does not follow blindly the prediction. 
Infeasibility of the prediction was detected at 82\% of the input sequence for $\epsilon=0.01$ and at 64\% for $\epsilon=0.1$.
As expected, with the performance degrades with the perturbation of the prediction.  
Note that the observed robustness is not monotone in $\eta$, unlike the bound shown in this paper.  We think that this performance degradation for $\eta$ around $1/2$ is due to the rather simplistic mixture between the primal-dual solution and the predicted solution.

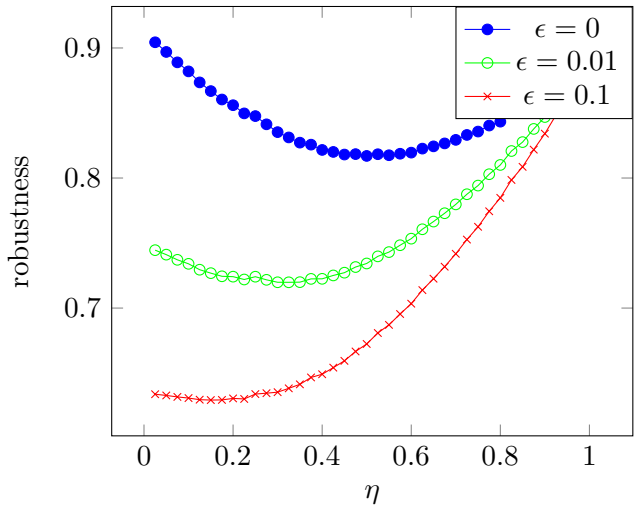
\begin{figure}[ht]
\begin{center}
	\begin{tikzpicture}[scale=1]
		\begin{axis}[
xlabel={$\eta$},
ylabel={robustness},
yticklabel style={/pgf/number format/precision=5},
legend style={at={(1,1)},anchor=north east},
legend entries={$\epsilon=0$,$\epsilon=0.01$,$\epsilon=0.1$}
]
\addplot [blue,mark=*] table {ratio_0.0.dat};
\addplot [green,mark=o] table {ratio_0.01.dat};
\addplot [red,mark=x] table {ratio_0.1.dat};
\end{axis}
	\end{tikzpicture}
\end{center}
\caption{Experimental evaluation}
\label{fig:experiments}
\end{figure}

\section{Conclusion}
%In this paper, we have presented primal-dual approaches based on configuration 
%linear programs to design competitive algorithms for covering problems with non-convex objectives.
%Non-convexity until now is considered as a strong barrier in optimization. 
%We hope that our approach would contribute some elements toward the study of non-linear/non-convex problems.

In the paper, we have presented primal-dual 
frameworks to design algorithms with predictions for non-linear problems with packing constraints.
Through applications, we show the potential of our approach and provide useful ideas/guarantees 
in incorporating predictions into solutions to problems with high impact such as submodular optimization and ad-auctions. 
An interesting direction is to prove lower bounds for non-linear packing problems
in terms of smoothness and confidence parameters.

\bibliographystyle{plainnat}
\bibliography{configuration}

\end{document}